\documentclass{article}
\usepackage[utf8]{inputenc}
\pdfoutput=1

\usepackage{fullpage}
\usepackage{amsmath}
\usepackage{algorithmic}

\title{Profile-Based Privacy for Locally Private Computations}
\author{Joseph Geumlek and Kamalika Chaudhuri }

\usepackage{color}
\usepackage{amsthm}
\usepackage{subfigure}

\newtheorem{Thm}{Theorem}
\newtheorem{theorem}{Theorem}
\newtheorem{Defn}[theorem]{Definition}
\newtheorem{Obs}[theorem]{Observation}
\usepackage{amssymb}
\usepackage{graphicx}
\usepackage{algorithm}
\usepackage{algorithmic}
\usepackage{graphicx}

\begin{document}

\maketitle

\begin{abstract}
Differential privacy has emerged as a gold standard in privacy-preserving data analysis. A popular variant is local differential privacy, where the data holder is the trusted curator. A major barrier, however, towards a wider adoption of this model is that it offers a poor privacy-utility trade-off.

In this work, we address this problem by introducing a new variant of local privacy called {\em{profile-based privacy}}. The central idea is that the problem setting comes with a graph $G$ of data generating distributions, whose edges encode sensitive pairs of distributions that should be made indistinguishable. This provides higher utility because unlike local differential privacy, we no longer need to make every pair of private values in the domain indistinguishable, and instead only protect the identity of the underlying distribution. We establish privacy properties of the profile-based privacy definition, such as post-processing invariance and graceful composition. Finally, we provide mechanisms that are private in this framework, and show via simulations that they achieve higher utility than the corresponding local differential privacy mechanisms.

\end{abstract}

\section{Introduction}
    A great deal of machine learning in the 21st century is carried out on sensitive data, and hence the field of privacy preserving data analysis is of increasing importance.  Differential privacy~\cite{DMNS06}, introduced in 2006, has become the dominant paradigm for specifying data privacy. A body of compelling results~\cite{CMS11, CSS12, KST12, FGWC16, Wang15} have been achieved in the "centralized" model, in which a trusted data curator has raw access to the data while performing the privacy-preserving operations. However, such trust is not always easy to achieve, especially when the trust must also extend to all future uses of the data.

%    An implementation of differential privacy that has been particularly popular in industrial applications makes each user into their own trusted curator. Commonly referred to as {\em{Local Differential Privacy}}~\cite{duchi2013local}, this model consists of users locally privatizing their own data before submission to an aggregate data curator. Due to the strong robustness of differential privacy under further computations, this model preserves privacy regardless of the trust in the aggregate curator, now or in the future. Two popular industrial systems implementing local differential privacy include Google's RAPPOR and Apple's iOS data collection systems.

    An implementation of differential privacy that has been particularly popular in industrial applications makes each user into their own trusted curator. Commonly referred to as {\em{local differential privacy}}~\cite{duchi2013local}, this model consists of users locally privatizing their own data before submission to an aggregate data curator. Due to the strong robustness of differential privacy under further computations, this model preserves privacy regardless of the trust in the aggregate curator, now or in the future. Two popular industrial systems implementing local differential privacy include Google's RAPPOR~\cite{erlingsson2014rappor} and Apple's iOS data collection systems.  
    
However, a major barrier for the local model is the undesirable utility sacrifices of the submitted data. A local differential privacy implementation achieves much lower utility than a similar method that assumes trusts in the curator. Strong lower bounds have been found for the local framework~\cite{duchi2013}, leading to pessimistic results requiring massive amounts of data to achieve both privacy and utility.

In this work, we address this challenge by proposing a new restricted privacy definition, called profile-based privacy. The central idea relies on specifying a graph $G$ of {\em{profiles}} or data generating distributions, where edges encode sensitive pairs of distributions that should be made indistinguishable. Our framework does not require that all features of the observed data be obscured; instead only the information connected to identifying the distributions must be perturbed. This offers privacy by obscuring data from sensitive pairs of profiles while side-stepping the utility costs of local differential privacy, where every possible pair of observations must be indistinguishable. As a concrete example, suppose we have commute trajectories from a city center to four locations - $A$, $B$, $C$ and $D$, where $A$ and $B$ are north and $C$ and $D$ are south of the center. Profile based privacy can make the trajectories that originate in $A$ and $B$ and those that originate in $C$ and $D$ indistinguishable. This offers better utility than local differential privacy, which would make every trajectory indistinguishable, while still offering finer grained privacy in the sense that an adversary will only be able to tell that the commuter began north or south of the city center.

We show that profile-based privacy satisfies some of the beneficial properties of differential privacy, such as post-processing invariance and certain forms of graceful composition. We provide new mechanisms in this definition that offer better utility than local differential privacy, and conclude with theoretical as well as empirical evidence of their effectiveness.

\section{Preliminary: Privacy definitions}

We begin with defining local differential privacy -- a prior privacy framework that is related to our definition. 

\begin{Defn}\label{def:dp}
A randomized mechanism $\mathcal{A} : \mathcal{X} \rightarrow \mathcal{Y}$ achieves $\epsilon$-local differential privacy if for every pair $(X,X')$ of individuals' private records in $\mathcal{X}$ and for all outputs $y \in \mathcal{Y}$ we have:

\begin{equation}
\Pr(\mathcal{A}(X) = y) \leq e^{\epsilon} \Pr(\mathcal{A}(X') = y) \mbox{.} \label{privconstraint}
\end{equation}
\end{Defn}

Concretely, local differential privacy limits the ability of an adversary to increase their confidence in whether an individual's private value is $X$ versus $X'$ even with arbitrary prior knowledge. These protections are robust to any further computation performed on the mechanism output.

\section{Profile-based Privacy Definition} \label{sec:pbpdef}

Before we present the definition and discuss its implications, it is helpful to have a specific problem in mind. We present one possible setting in which our profiles have a clear interpretation.

\subsection{Example: Resource Usage Problem Setting}

%    Imagine a shared workstation with access to several resources, such as network bandwidth, specialized hardware, or electricity usage. Different users might use this workstation, coming from a diverse pool of job titles and roles. An analyst wishes to collect and analyze the metrics of resource usage, but the analyst also wishes to respect the privacy of the workstation users, but the choice of privacy framework remains open. At one extreme is local differentially privacy, in which every value of a resource usage metric is considered sensitive and privatized. Under the profile-based framework, a choice exists to select the desired information protected. If every user has their own unique resource usage profile, a profile-based implementation can emit resource usage while protecting the precise identity of the workstation user. In the absence of user-specific profile information, an implentation assigning usage profiles each distinct job role could still achieve meaningful privacy protections.

 Imagine a shared workstation with access to several resources, such as network bandwidth, specialized hardware, or electricity usage. Different users might use this workstation, coming from a diverse pool of job titles and roles. An analyst wishes to collect and analyze the metrics of resource usage, but also wishes to respect the privacy of the workstation users. With local differential privacy, any two distinct measurements must be rendered indistinguishable. Under our alternative profile-based framework, a choice exists to protect user identities instead of measurement values. This shifts the goal away from hiding all features of the resource usages, and permits measurements to be released more faithfully when not indicative of a user's identity.

%In the absence of user-specific profile information, an implentation assigning usage profiles each distinct job role could still achieve meaningful privacy protections.
    
\subsection{Definition of Profile-based Differential Privacy}

Our privacy definition revolves around a notion of profiles, which represent distinct potential data-generating distributions. To preserve privacy, the mechanism's release must not give too much of an advantage in guessing the release's underlying profile. However, other facets of the observed data can (and should) be preserved, permitting greater utility than local differential privacy.

\begin{Defn}
Given a graph $G = (\mathcal{P},E)$ consisting of a collection $\mathcal{P}$ of data-generating distributions ("profiles") over the space $\mathcal{X}$ and collection of edges $E$, a randomized mechanism $\mathcal{A} : \mathcal{X} \times \mathcal{P} \rightarrow \mathcal{Y}$ achieves $(G,\epsilon)$-profile-based differential privacy if for every edge $e \in E$ connecting profiles $P_i$ and $P_j$, with random variables $X_i \sim P_i$ and $X_j \sim P_j$, and for all outputs $y \in \mathcal{Y}$ we have:

\begin{equation}
\frac{\Pr(\mathcal{A}(X_i,P_i) = y )}{\Pr(\mathcal{A}(X_j,P_j) = y)} \leq e^\epsilon\mbox{.} \label{profprivconstraint}
\end{equation}

\end{Defn}

 Inherent in this definition is an assumption on adversarial prior knowledge: the adversary knows each profile distribution, but has no further auxiliary information about the observed data $X$. The protected secrets are the identities of the source distributions, and are not directly related to particular features of the data $X$. Put another way, the adversarial goal here is to distinguish $P_i$ versus $P_j$, rather than any fixed $X$ versus $X'$ pair in local differential privacy. These additional assumptions in the problem setting, however, permit better performance. By not attempting to completely privatize the raw observations, information that is less relevant for guessing the sensitive profile identity can be preserved for downstream tasks.
 
 The flexible specification of sensitive pairs via edges in the graph permits privacy design decisions that also impact the privacy-utility trade-off. When particular profile pairs are declared less sensitive, the perturbations required to blur those profiles can be avoided. Such decisions would be impractical in the data-oriented local differential privacy setting, where the space of pairs of data sets is intractably large.
 
 \subsection{Related Work and Related Privacy Definitions}

Our proposed definition is related to two commonly used privacy frameworks: the generalized Pufferfish privacy framework~\cite{KM12}, and geoindistinguishability~\cite{geoind13}. Like our definition, Pufferfish presents an explicit separation of sensitive and insensitive information with distributional assumptions. However, we focus on a local case with distributional secrets, while the existing Pufferfish literature targets value-dependent secrets in a global setting. Our definition is also similar to geoindistinguishability, but our work does not require an explicit metric and applies more readily to a discrete setting. This view can also be extended to limited precision local privacy, a definition appearing in independent concurrent work on privatized posterior sampling~\cite{schein2018locally}.

Another closely related concurrent independent work is DistP~\cite{kawamoto2018differentially}. This work also defines a privacy notion over a set of adjacencies across data distributions. They focus on continuous settings, lifting the Laplace mechanism into an $(\epsilon,\delta)$-style privacy guarantee over distributions. Additionally, they explore extended privacy frameworks that intorduce metrics into the guarantees, much like geoindistinguishability. Our work is mainly complementary to \cite{kawamoto2018differentially}. We focus instead on discrete settings where the specific structure of our profile adjacency graph can be exploited, as well as finding mechanisms that achieve privacy guarantees without explicit metrics or additive $\delta$ terms.

Our methods also resemble those seen under maximal-leakage-constrained hypothesis testing~\cite{liao2017hypothesis}. The maximal-leakage framework also employs a distribution-focused mechanism design, but solves a different problem. Our work aims to prevent identifying distributions while preserving and identifying observed values where possible. The maximal-leakage setting inverts this goal, and protects the observed values while maximizing the detection of hypotheses on the distributions. This distinction in goal also applies with respect to the distributional privacy framework~\cite{balcan2012distributed} and in the explorations of \cite{canonne2018structure} into general privatized hypothesis testing. 

Finally, our work can also be viewed in relation to information theoretic definitions dues to the deep connections present from differential privacy~\cite{wang2016relation}. However, like with other differential privacy variants, the worst-case protections of the sensitive attributes makes our notion of privacy distinct from many common averaged measures of information.

\subsection{Discussion of the Resource Usage Problem}

This privacy framework is quite general, and as such it helps to discuss its meaning in more concrete terms. Let us return to the resource usage setting. We'll assume that each user has a personal resource usage profile known prior to the data collection process. The choice of edges in the graph $G$ has several implications. If the graph has many edges, the broad identity of the workstation user will be hidden by forbidding many potential inferences. However, this protection does not require all the information about resource usage to be obscured. For example, if all users require roughly the same amount of electricity at the workstation, then electrical usage metrics will not require much obfuscation even with a fully connected graph.

A more sparse graph might only connect profiles with the same job title or role. These sensitive pairs will prevent inferences about particular identities within each role. However, without connections across job titles, no protection is enforced against inferring the job title of the current workstation user. Thus such a graph declares user identities sensitive, while a user's role is not sensitive. When users in the same role have similar usages, this sparser graph will require less perturbations of the data.

One important caveat of this definition is that the profile distributions must be known and are assumed to be released a priori, i.e. they are not considered privacy sensitive. If the user profiles cannot all be released, this can be mitigated somewhat by reducing the granularity of the graph. A graph consisting only of profiles for each distinct job role can still encode meaningful protections, since limiting inferences on job role can also limit inferences on highly correlated information like the user's identity.

The trade-off in profile granularity is subtle, and is left for future exploration. More profiles permit more structure and opportunities for our definition to achieve better utility than local differential privacy, but also require a greater level of a priori knowledge.

\section{Properties}

Our privacy definition enjoys several similar properties to other differential-privacy-inspired frameworks. The post-processing and composition properties are recognized as highly desired traits for privacy definitions~\cite{KM12}.

\paragraph{Post-Processing}
The post-processing property specifies that any additional computation (without access to the private information) on the released output cannot result in worse privacy. Following standard techniques, our definition also shares this data processing inequality.

\begin{Obs}
If a data sample $X_i$ is drawn from profile $P_i$, and $\mathcal{A}$ preserves $(G,\epsilon)$-profile-based privacy, then for any (potentially randomized) function $F$, the release $F(\mathcal{A}(X_i,P_i))$ preserves $(G,\epsilon)$-profile-based privacy.
\end{Obs}

\paragraph{Composition} The composition property allows for multiple privatized releases to still offer privacy even when witnessed together. Our definition also gets a compositional property, although not all possible compositional settings behave nicely. We mitigate the need for composition by focusing on a local model where the data mainly undergoes one privatization. 

Profile-based differential privacy enjoys additive composition if truly independent samples $X$ are drawn from the same profile. The proof of this follows the same reasoning as the additive composition of differential privacy.

\begin{Obs}
If two independent samples $X_1$ and $X_2$ are drawn from profile $P_i$, and $\mathcal{A}_1$ preserves $(G,\epsilon_1)$-profile-based privacy and $\mathcal{A}_2$ preserves $(G,\epsilon_2)$-profile-based privacy, then the combined release $(\mathcal{A}_1(X_1,P_i) , \mathcal{A}_2(X_2,P_i))$ preserves $(G, \epsilon_1 + \epsilon_2)$-profile-based privacy.

\end{Obs}

A notion of parallel composition can also be applied if two data sets come from two independent processes of selecting a profile. In this setting, information about one instance has no bearing on the other. This matches the parallel composition of differential privacy when applied to multiple independent individuals, and would be the analagous setting to local differential privacy where each individual applies their own mechanism. 

\begin{Obs}\label{obs:parallelcomp}
If two profiles $P_i$ and $P_j$ are independently selected, and two observations $X_i \sim P_i$ and $X_j \sim P_j$ are drawn,  and $\mathcal{A}_1$ preserves $(G,\epsilon_1)$-profile-based privacy and $\mathcal{A}_2$ preserves $(G,\epsilon_2)$-profile-based privacy, then the combined release $(\mathcal{A}_1(X_i,P_i), \mathcal{A}_2(X_j,P_j))$ preserves $(G, \text{max}\{\epsilon_1,\epsilon_2\})$-profile-based privacy.

\end{Obs}

The parallel composition result assumes that the choice of $\mathcal{A}_2$ does depend on the first release, or in other words that it is non-adaptive. It should also be noted that the privacy guarantee is about how much protection a single profile graph edge receives in just one profile selection process. With two releases, clearly more information is being released, but the key idea in this result is that the information released in one round has no impact on the secret profile identity of the other round.

However, this framework cannot offer meaningful protections against adversaries that know about correlations in the profile selection process. For example, consider an adversary with knowledge that profile $P_k$ is always selected immediately after either $P_i$ or $P_j$ are selected. An edge obscuring $P_i$ versus $P_j$ will not prevent the adversary from deducing $P_k$ in the next round. This matches the failure of differential privacy to handle correlations across individuals. The definition also does not compose if the same observation $X$ is reprocessed, as it adds correlations unaccounted for in the privacy analysis. Although such compositions would be valuable, it is less important when the privatization occurs locally at the time of data collection.

% \begin{figure}
% \centering
% \includegraphics[width=0.5\textwidth]{compositionsettings.pdf}
% \caption{Summary of Composition Results. From left to right: independent observation samples with independent mechanism applications compose additively, independent profile selections compose in parallel, dependent observation samples from the same profile do not compose nicely, and dependent profile selections do not compose nicely.}
% \end{figure}

Placing these results in the context of reporting resource usage, we can bound the total privacy loss across multiple releases in two cases. Additive composition applies if a single user emits multiple independent measurements and each measurement is separately privatized. When two users independently release measurements, each has no bearing on the other and parallel composition applies. If correlations exist across measurements (or across the selection of users), no compositional result is provided.

\begin{figure}
\centering
\includegraphics[width=5in]{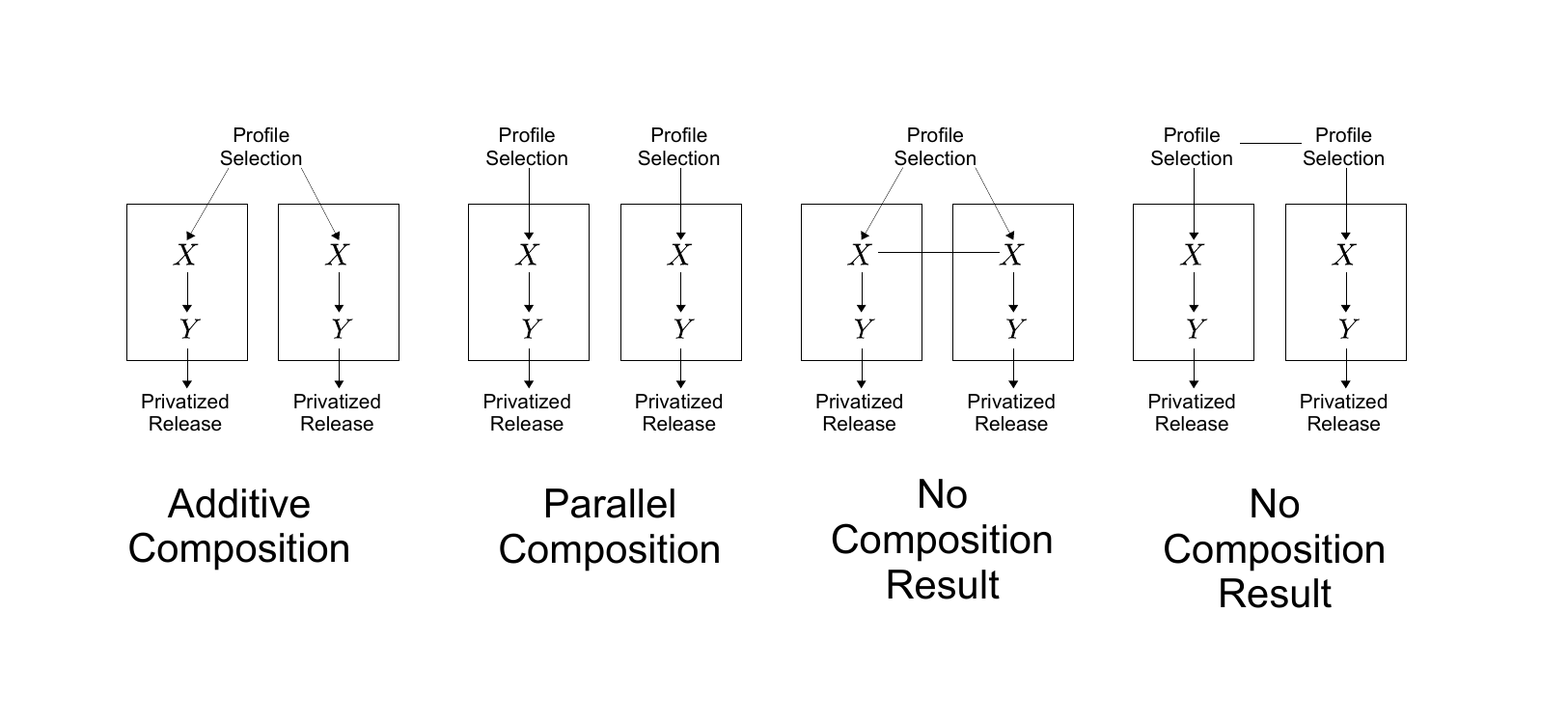}
\caption{Summary of Composition Results. From left to right: independent observation samples with independent mechanism applications compose additively, independent profile selections compose in parallel, dependent observation samples from the same profile do not compose nicely, and dependent profile selections do not compose nicely.}
\end{figure}

\section{Mechanisms} 

We now provide mechanisms to implement the profile-based privacy definition. Before getting into specifics, let us first consider the kind of utility goals that we can hope to achieve. We have two primary aspects of the graph $G$ we wish to exploit. First, we wish to preserve any information in the input that does not significantly indicate profile identities. Second, we wish to use the structure of the graph and recognize that some regions of the graph might require less perturbations than others.

\subsection{The One-Bit Setting}

We begin with a one-bit setting -- where the input to the mechanism is a single private bit -- and build up to the more general discrete setting. 

%and build up to a more complicated mechanism with fewer restrictions. 

The simplest case is when we have two profiles $i$ and $j$ represented by Bernoulli distributions $P_i$ and $P_j$ with parameters $p_i$ and $p_j$ respectively. Here, we aim to design a mechanism $\mathcal{A}$ that makes a bit $b$ drawn from $P_i$ or $P_j$ indistinguishable; that is, for any $t \in \{0, 1\}$, with $b_i \sim P_i$ and $b_j \sim P_j$,
\begin{equation}
\frac{\Pr(\mathcal{A}(b_i,P_i) = t)}{\Pr(\mathcal{A}(b_j,P_j) = t )} \leq e^{\epsilon} \mbox{.}
\end{equation}

A plausible mechanism is to draw a bit $b'$ from a Bernoulli distribution that is independent of the original bit $b$. However, this is not desirable as the output bit would lose any correlation with the input, and any and all information in the bit $b$ would be discarded.

We instead use a mechanism that flips the input bit with some probability $\alpha \leq 1/2$. Lower values of $\alpha$ improve the correlation between the output and the input. The flip-probability $\alpha$ is obtained by solving the following optimization problem:

\begin{eqnarray}
\min && \alpha  \label{eq:onebitprob} \\
{\text{subject\ to }} && \alpha \geq 0 \nonumber\\
&&\frac{p_i (1 - \alpha) + (1 - p_i) \alpha}{p_j (1 - \alpha) + (1 - p_j) \alpha} \in [e^{-\epsilon}, e^{\epsilon}] \nonumber\\
&&\frac{(1 - p_i) (1 - \alpha) + p_i \alpha }{(1 - p_j) (1 - \alpha) + p_j \alpha} \in [e^{-\epsilon}, e^{\epsilon}] \mbox{.}\nonumber
\end{eqnarray}

When $p_i = 0$ and $p_j = 1$ (or vice versa), this reduces to the standard randomized response mechanism~\cite{W65}; however, $\alpha$ may be lower if $p_i$ and $p_j$ are closer -- a situation where our utility is better than local differential privacy's.

\begin{algorithm} 
\caption{Single-bit Two-profile Mechanism}\label{alg:onebittwoprof}
\begin{algorithmic}[1]
\REQUIRE Two Bernoulli profiles parameterized by $p_1$ and $p_2$, privacy level $\epsilon$, input profile $P_i$, input bit $x$.
\STATE Solve the linearly constrained optimization \eqref{eq:onebitprob} to get a flipping probability $\alpha$.
\STATE With probability $\alpha$, return the negation of $x$. Otherwise, return $x$.
\end{algorithmic}
\end{algorithm}

The mechanism described above only addresses two profiles. If we have a cluster of profiles representing a connected component of the profile graph, we can compute the necessary flipping probabilities across all edges in the cluster. To satisfy all the privacy constraints, it suffices to always use a flipping probability equal to the largest value required by an edge in the cluster. This results in a naive method we will call the One Bit Cluster mechanism, directly achieves profile-based privacy. 

\begin{algorithm} 
\caption{One Bit Cluster Mechanism}\label{alg:onebit}
\begin{algorithmic}[1]
\REQUIRE Graph $(\mathcal{P},E)$ of Bernoulli profiles, privacy level $\epsilon$, input profile $P_i$, input bit $x$.
\FOR{each edge $e$ in the connected component of the graph containing $P_i$}
\STATE Solve the linearly constrained optimization \eqref{eq:onebitprob} to get a flipping probability $\alpha_e$ for this edge.
\ENDFOR
\STATE Compute $\alpha = \max_{e} \alpha_e$.
\STATE With probability $\alpha$, return the negation of $x$. Otherwise, return $x$.
\end{algorithmic}
\end{algorithm}

\begin{theorem}
The One Bit Cluster mechanism achieves $(G,\epsilon)$-profile-based privacy.
\end{theorem}

The One Bit Cluster mechanism has two limitations. First, it applies only to single bit settings and Bernoulli profiles, and not categorical distributions. Second, by treating all pairs of path-connected profiles similarly, it is overly conservative; when profiles are distant in the graph from a costly edge, it is generally possible to provide privacy with lesser perturbations for these distant profiles.

We address the second drawback while remaining in the one bit setting with the Smooth One Bit mechanism, which uses ideas inspired by the smoothed sensitivity mechanism in differential privacy~\cite{NRS07}. However, rather than smoothly calibrating the perturbations across the entire space of data sets, a profile-based privacy mechanism needs only to smoothly calibrate over the specified profile graph. This presents a far more tractable task than smoothly handling all possible data sets in differential privacy.

This involves additional optimization variables, $\alpha_1,\ldots,\alpha_k$, for each of the $k$ profiles in $G$. Thus each profile is permitted its own chance of inverting the released bit.  Here, $p_i$ once again refers to the parameter of the Bernoulli distribution $P_i$. We select our objective function as $\max(\alpha_1,\ldots,\alpha_k)$ in order to uniformly bound the mechanism's chances of inverting or corrupting the input bit. This task remains convex as before, and is still tractably optimized.

\begin{eqnarray}
\min_{\alpha_1,\ldots,\alpha_k} && \max(\alpha_1,\ldots,\alpha_k) \label{eq:onebitsmooth} \\
{\text{subject\ to }} && \forall i \in [k]\text{: }\alpha_i \geq 0 \nonumber\\
&&\frac{p_i (1 - \alpha_i) + (1 - p_i) \alpha_i}{p_j (1 - \alpha_j) + (1 - p_j) \alpha_j} \in [e^{-\epsilon}, e^{\epsilon}] \nonumber\\
&&\frac{(1 - p_i) (1 - \alpha_i) + p_i \alpha_i }{(1 - p_j) (1 - \alpha_j) + p_j \alpha_j} \in [e^{-\epsilon}, e^{\epsilon}] \mbox{.}\nonumber
\end{eqnarray}

\begin{algorithm} 
\caption{Smooth One Bit Mechanism}\label{alg:smooth1bit}
\begin{algorithmic}[1]
\REQUIRE Graph $(\mathcal{P},E)$ of $k$ Bernoulli profiles, privacy level $\epsilon$, input profile $P_i$, input bit $x$.
\STATE Solve the linearly constrained optimization \eqref{eq:onebitsmooth} to get flipping probabilities $\alpha_1,\ldots,\alpha_k$.
\STATE With probability $\alpha_i$, return the negation of $x$. Otherwise, return $x$.
\end{algorithmic}
\end{algorithm}

\begin{theorem}
The Smooth One Bit mechanism achieves $(G,\epsilon)$-profile-based privacy.
\end{theorem}

%This mechanism however is still too simple to exploit either of the utility goals of Section \ref{sec:utilitygoals}. With 1-dimensional Bernoulli profiles, a single bit cannot hold any orthogonal information. It also fails to achieve smooth sensitivity, since the One Bit Cluster mechanism restricts itself to using a single flipping probability across the entire connected component.

%We can address the smooth sensitivity by introducing additional optimization variables to assign distinct flipping chances to each profile. Only the profiles that require heavy perturbation will have large flipping probabilities. We can similarly address the orthogonal information goal by moving to a categorical setting. Each potential observation of each profile receives a distinct distribution over the output space, and once again we optimize these additional variables subject to our privacy constraints. In this higher dimensional settings, some dimensions might be orthogonal to the profile choice. For brevity we will perform both modifications at once.

\subsection{The Categorical Setting}

We now show how to generalize this model into the categorical setting. This involves additional constraints, as well as a (possibly) domain specific objective that maximizes some measure of fidelity between the input and the output.

Specifically, suppose we have $k$ categorical profiles each with $d$ categories; we introduce $kd^2$ variables to optimize, with each profile receiving a $d \times d$ transition matrix. To keep track of these variables, we introduce the following notation:

\begin{itemize}
\item $P_i,\ldots,P_k$: a set of $k$ categorical profiles in $d$ dimensions encoded as a vector.
\item $A^{1},\ldots,A^k$: A set of  $d$-by-$d$ transition matrix that represents the mechanism's release probabilities for profile $i$. $A^i_{j,k}$ represents the $(j,k)$-th element of the matrix $A^i$.
\item $P_iA^i$ represents the $d$ dimensional categorical distribution induced by the transition matrix $A^i$ applied to the distribution $P_i$.
\item In an abuse of notation, $P_iA^i \leq e^{\epsilon} P_jA^j$ is a constraint that applies element-wise to all components of the resulting vectors on each side.
\end{itemize}

With this notation, we can express our optimization task:

\begin{align}
\min_{A^{1},\ldots,A^k}  &\text{max}(\text{off-diagonal entries of } A^1,\ldots,A^k) \label{eq:smoothobj}\\
\text{subject to }   &\forall i \in [n] \forall j \in [d] \forall k \in [d] \text{: } \hskip10pt  0 \leq A^i_{j,k} \leq 1 \nonumber\\
 &\forall i \in [n] \forall j \in [d] \text{: } \hskip10pt  \sum_{k=1}^d A^i_{j,k} = 1 \nonumber\\
  &\forall (P_i,P_j) \in E \text{: } P_iA^i \leq e^{\epsilon} P_jA^j, \hskip10pt   P_jA^j \leq e^{\epsilon} P_iA^i \mbox{.}\nonumber
\end{align}

\begin{algorithm} 
\caption{Smooth Categorical Mechanism}\label{alg:smoothcat}
\begin{algorithmic}[1]
\REQUIRE Graph $(\mathcal{P},E)$ of Categorical profiles, privacy level $\epsilon$, input profile $P_i$, input $x$.
\STATE Solve the linearly constrained optimization \eqref{eq:smoothobj} to get the transition matrices $A^{1},\ldots,A^k$.
\STATE Return $y$ sampled according to the $x$th row of $A^i$.
\end{algorithmic}
\end{algorithm}

To address the tractability of the optimization, we note that each of the privacy constraints are linear constraints over our optimization variables. We further know the feasible solution set is nonempty, as trivial non-informative mechanisms achieve privacy. All that is left is to choose a suitable objective function to make this a readily solved convex problem.

To settle onto an objective will require some domain-specific knowledge of the trade-offs between choosing which profiles and which categories to report more faithfully. Our general choice is a maximum across the off-diagonal elements, which attempts to uniformly minimize the probability of any data corruptions. This can be further refined in the presence of a prior distribution over profiles, to give more importance to the profiles more likely to be used.

We define the Smooth Categorical mechanism as the process that solves the optimization~\eqref{eq:smoothobj} and applies the appropriate transition probabilities on the observed input.

\begin{theorem}
The Smooth Categorical mechanism achieves $(G,\epsilon)$-profile-based privacy.
\end{theorem}

\subsection{Utility Results}

 The following results present utility bounds which illustrate potential improvements upon local differential privacy; a more detailed numerical simulation is presented in Section~\ref{sec:experiments}.
 
 \begin{theorem}
 If $\mathcal{A}$ is a mechanism that preserves $\epsilon$-local differential privacy, then for any graph $G$ of sensitive profiles, $\mathcal{A}$ also preserves $(G,\epsilon)$-profile-based differential privacy.
 \label{thm:noworse}
 \end{theorem}
 
 An immediate result of Theorem~\ref{thm:noworse} is that, in general and for any measure of utility on mechanisms, the profile-based differential privacy framework will never require worse utility than a local differential privacy approach. However, in specific cases, stronger results can be shown.

\begin{Obs} \label{thm:utility}
Suppose we are in the single-bit setting with two Bernoulli profiles $P_i$ and $P_j$ with parameters $p_i$ and $p_j$ respectively. If $p_i \leq p_j \leq e^{\epsilon} p_j$, then the solution $\alpha$ to~\eqref{eq:onebitprob} satisfies 
\begin{equation}
\alpha \leq\max\{0,\frac{p_j - e^{\epsilon}p_i}{2(p_j - e^{\epsilon}p_i) - (1-e^{\epsilon})}, \frac{p_i - e^{\epsilon} p_j + e^{\epsilon} - 1}{2(p_i - e^{\epsilon} p_j) +  e^{\epsilon} - 1}\}\mbox{.}
\end{equation}

\end{Obs}

Observe that to attain local differential privacy with parameter $\epsilon$ by a similar bit-flipping mechanism, we need a flipping probability of $\frac{1}{1 + e^{\epsilon}} = \frac{1}{1 + (1+ e^{\epsilon} - 1)}$, while we get bounds of the form $\frac{1}{1+(1+\frac{e^{\epsilon}-1}{p_j-e^{\epsilon}p_i})}$. Thus, profile based privacy does improve over local differential privacy in this simple case. The proof of Observation~\ref{thm:utility} follows from observing that this value of $\alpha$ satisfies all constraints in the optimization problem~\eqref{eq:onebitprob}.

\section{Evaluation} \label{sec:experiments}

We next evaluate our privacy mechanisms and compare them against each other and the corresponding local differential privacy alternatives. In order to understand the privacy-utility trade-off unconfounded by model specification issues, we consider synthetic data in this paper.

\subsection{Experimental Setup}

We look at three experimental settings  -- Bernoulli-Couplet, Bernoulli-Chain and Categorical-Chain-3. 

\paragraph{Settings.} In Bernoulli-Couplet, the profile graph consists of two nodes connected by a single edge $G = ( \mathcal{P} = \{a, b \}, E = \{ (a, b) \})$. Additionally, each profile is a Bernoulli distribution with a parameter $p$.

In Bernoulli-Chain, the profile graph consists of a {\em{chain}} of nodes, where successive nodes in the chain are connected by an edge. Each profile is still a Bernoulli distribution with parameter $p$. We consider two experiments in this category -- Bernoulli-Chain-6, where there are six profiles corresponding to six values of $p$ that are uniformly spaced across the interval $[0,1]$, and Bernoulli-Chain-21, where there are $21$ profiles corresponding to $p$ uniformly spaced on $[0, 1]$.

Finally, in Categorical-Chain, the profile graph comprises of three nodes connected into a chain $P_1-P_2-P_3$. Each profile however, corresponds to a 4-dimensional categorical distribution, instead of Bernoulli.

\begin{table}[htbp]
%   % increase table row spacing, adjust to taste
  \renewcommand{\arraystretch}{1}
  \caption{Categorical-Chain Profiles}
  \label{tab:table_example}
  \centering
  % Some packages, such as MDW tools, offer better commands for making tables
  % than the plain LaTeX2e tabular which is used here.
  \begin{tabular}{ l | c c c c }
		$P_1$ & 0.2 & 0.3 & 0.4 & 0.1 \\
		$P_2$ & 0.3 & 0.3 & 0.3 & 0.1 \\
		$P_3$ & 0.4 & 0.4 & 0.1 & 0.1 \\
		\end{tabular}
\end{table}

\paragraph{Baselines.} For Bernoulli-Couplet and Bernoulli-Chain, we use Warner's Randomized Response mechanism~\cite{W65} as a local differentially private baseline. For Categorical-Chain, the corresponding baseline is the $K$-ary version of randomized response. 

For Bernoulli-Couplet, we use our Smooth One Bit mechanism to evaluate our framework. For Categorical-Chain, we use the Smooth Categorical mechanism. 

\subsection{Results}

Figure~\ref{subfig:experiments1} plots the flipping probability for Bernoulli-Couplet as a function of the difference between profile parameters $p$. We find that as expected, as the difference between the profile parameters grows, so does the flipping probability and hence the noise added. However, in all cases, this probability stays below the corresponding value for local differential privacy -- the horizontal black line -- thus showing that profile-based privacy is an improvement. 

Figure~\ref{fig:experiments2} plots the probability that the output is $1$ as a function of $\epsilon$ for each profile in Bernoulli-Chain-6 and Bernoulli-Chain-21. We find that as expected for low $\epsilon$, the probability that the output is $1$ is close to $1/2$ for both the local differential privacy baseline and our method, whereas for higher $\epsilon$, it is spread out more evenly, (which indicates higher correlation with the input and higher utility). Additionally, we find that our Smooth One Bit mechanism  performs better than the baseline in both cases. 

\begin{figure}[tbp]
  \centerline{\subfigure[Bernoulli-Couplet, Our Method and Baseline.]{\includegraphics[width=2in,height=1.5in]{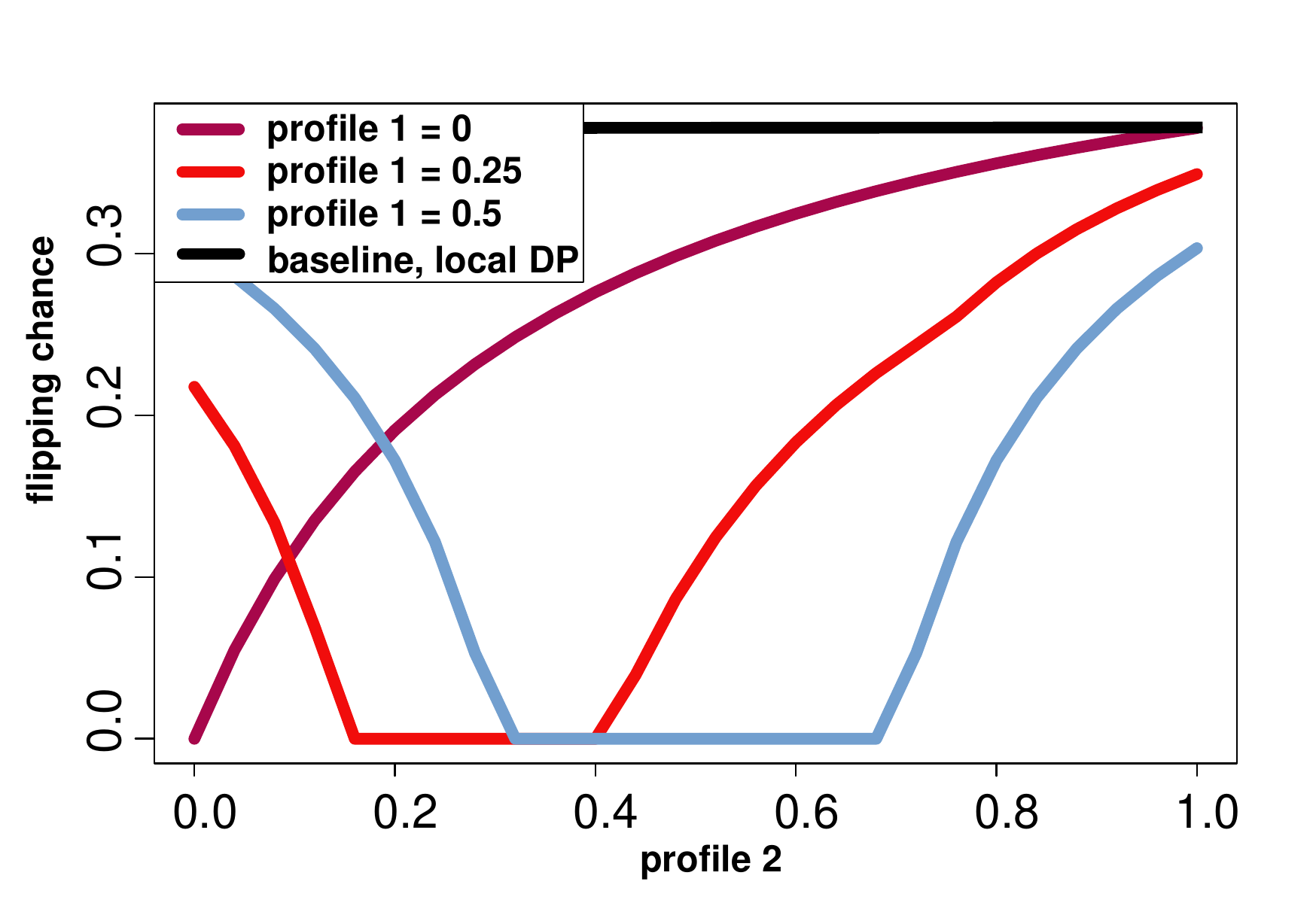}
      % where an .eps filename suffix will be assumed under latex,
      % and a .pdf suffix will be assumed for pdflatex
      \label{subfig:experiments1}}
      \hfil
  \subfigure[Categorical-Chain, Baseline (Local differential privacy). All 4 curves overlap.]{\includegraphics[width=2in,height=1.5in]{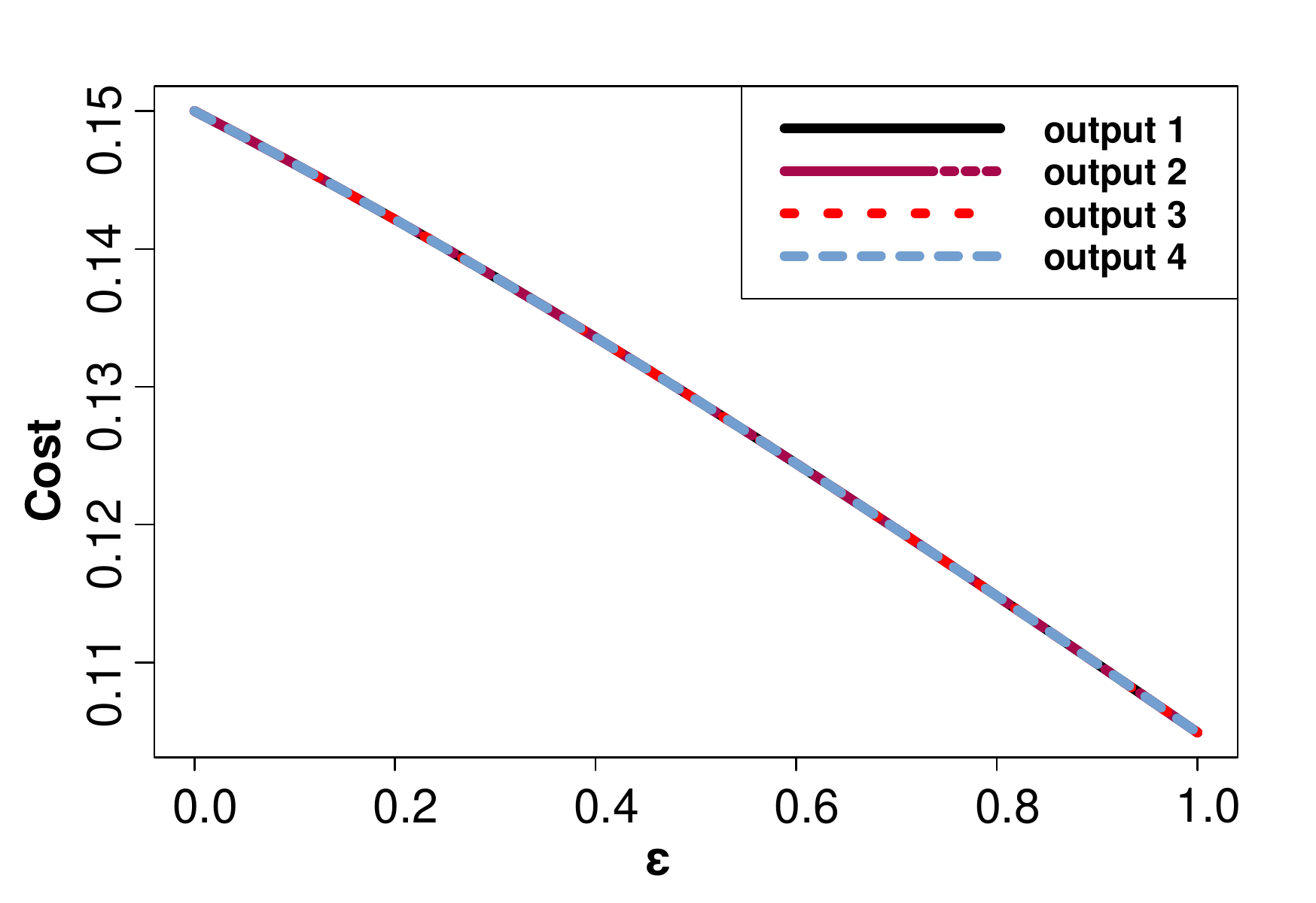}
      % where an .eps filename suffix will be assumed under latex,
      % and a .pdf suffix will be assumed for pdflatex
      \label{subfig:catbaseline}}
    \hfil
    \subfigure[Categorical-Chain, Our Method.]{\includegraphics[width=2in,height=1.5in]{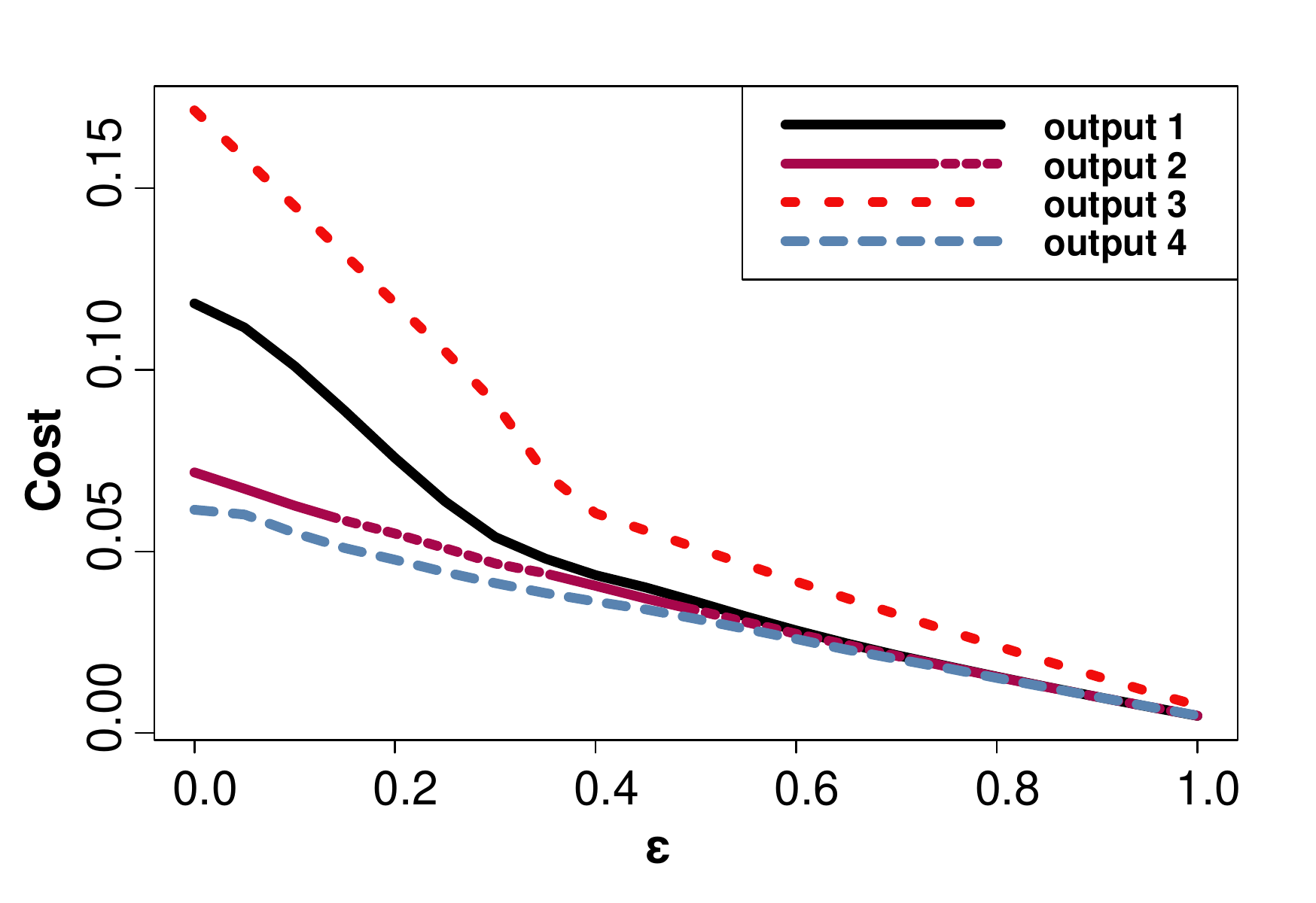}
      % where an .eps filename suffix will be assumed under latex,
      % and a .pdf suffix will be assumed for pdflatex
      \label{subfig:catourmethod}}
      }
  \caption{Experimental results in various settings. In all figures, lower is better.}
  \label{fig:experiments2}
\end{figure}

\begin{figure}
	\centerline{
	\subfigure[Bernoulli-Chain-6, Baseline.]{
            \includegraphics[width=3in,height=3in]{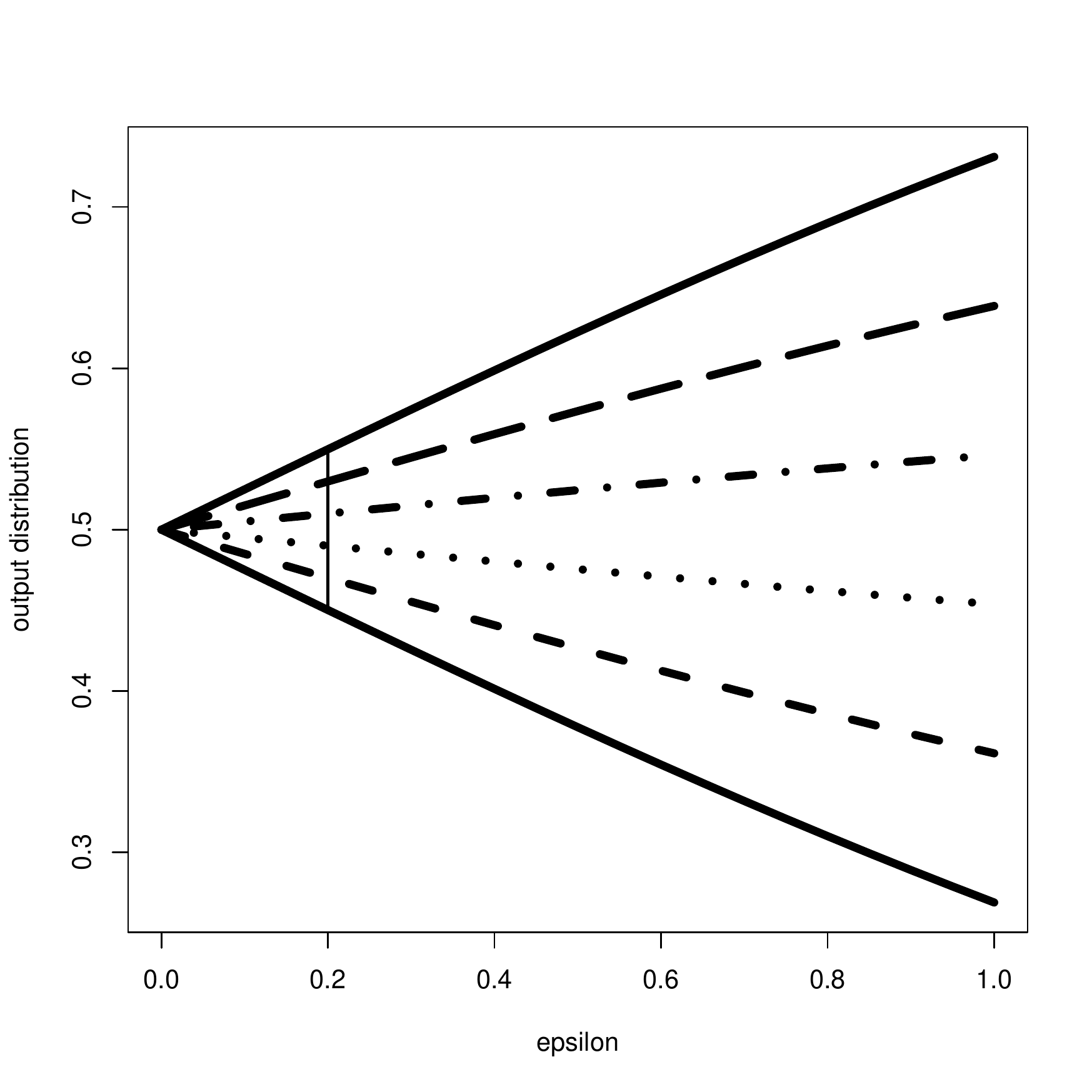}
	}
	\hfil
	\subfigure[Bernoulli-Chain-21, Baseline.]{
		\includegraphics[width=3in,height=3in]{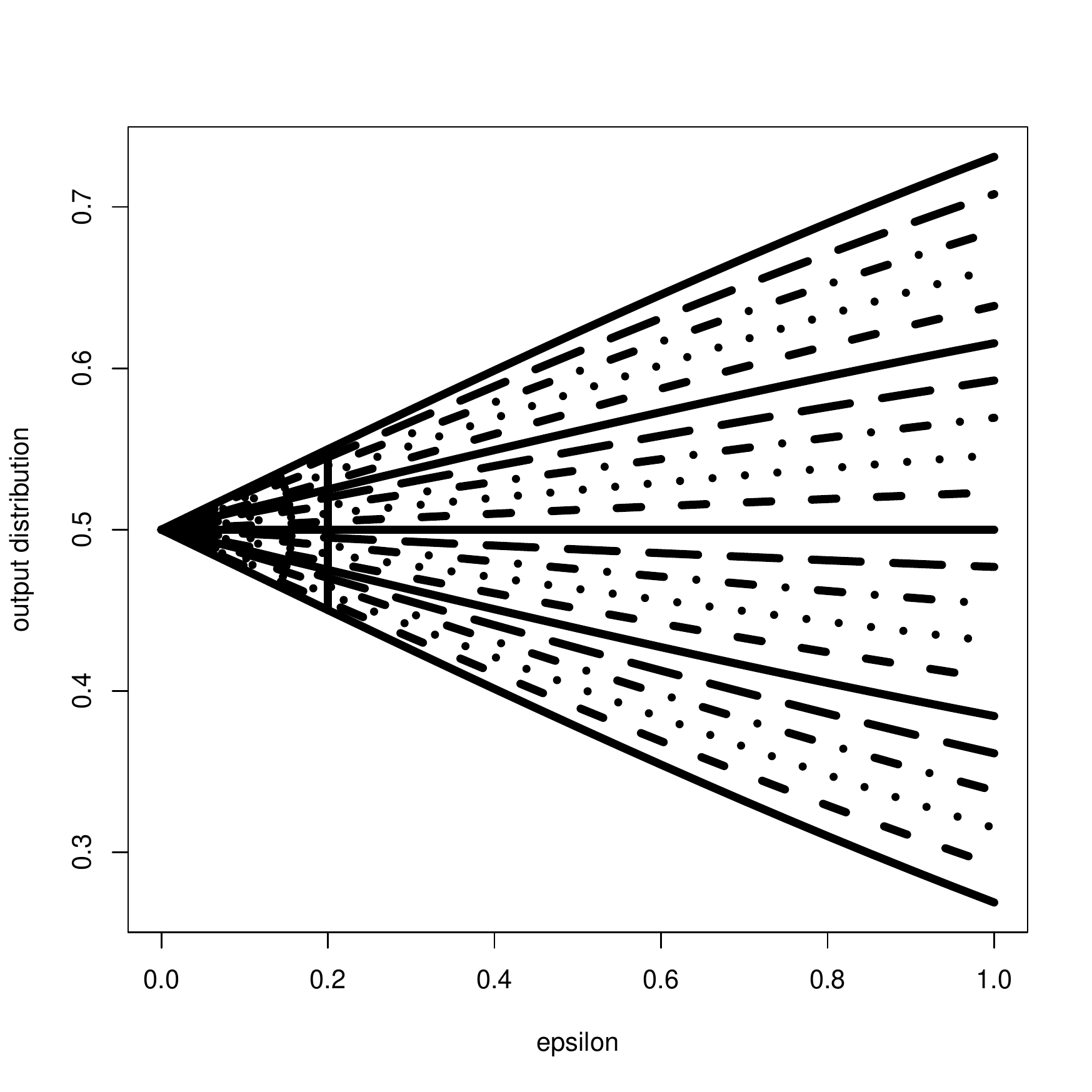}
	}}
	\centerline{
	\subfigure[Bernoulli-Chain-6, Our Method.]{
		\includegraphics[width=3in,height=3in]{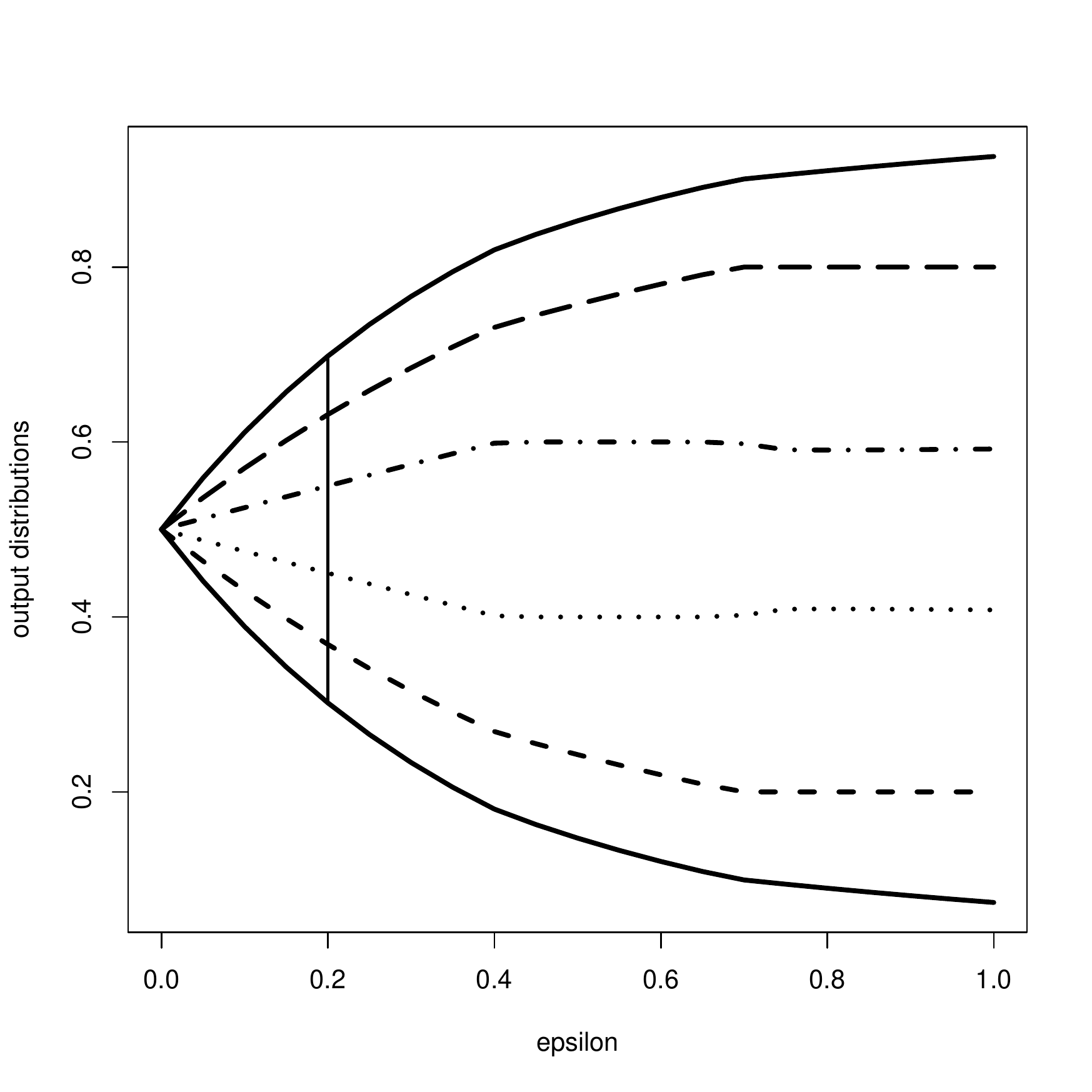}
	}
	\hfil
	\subfigure[Bernoulli-Chain-21, Our Method.]{
		\includegraphics[width=3in,height=3in]{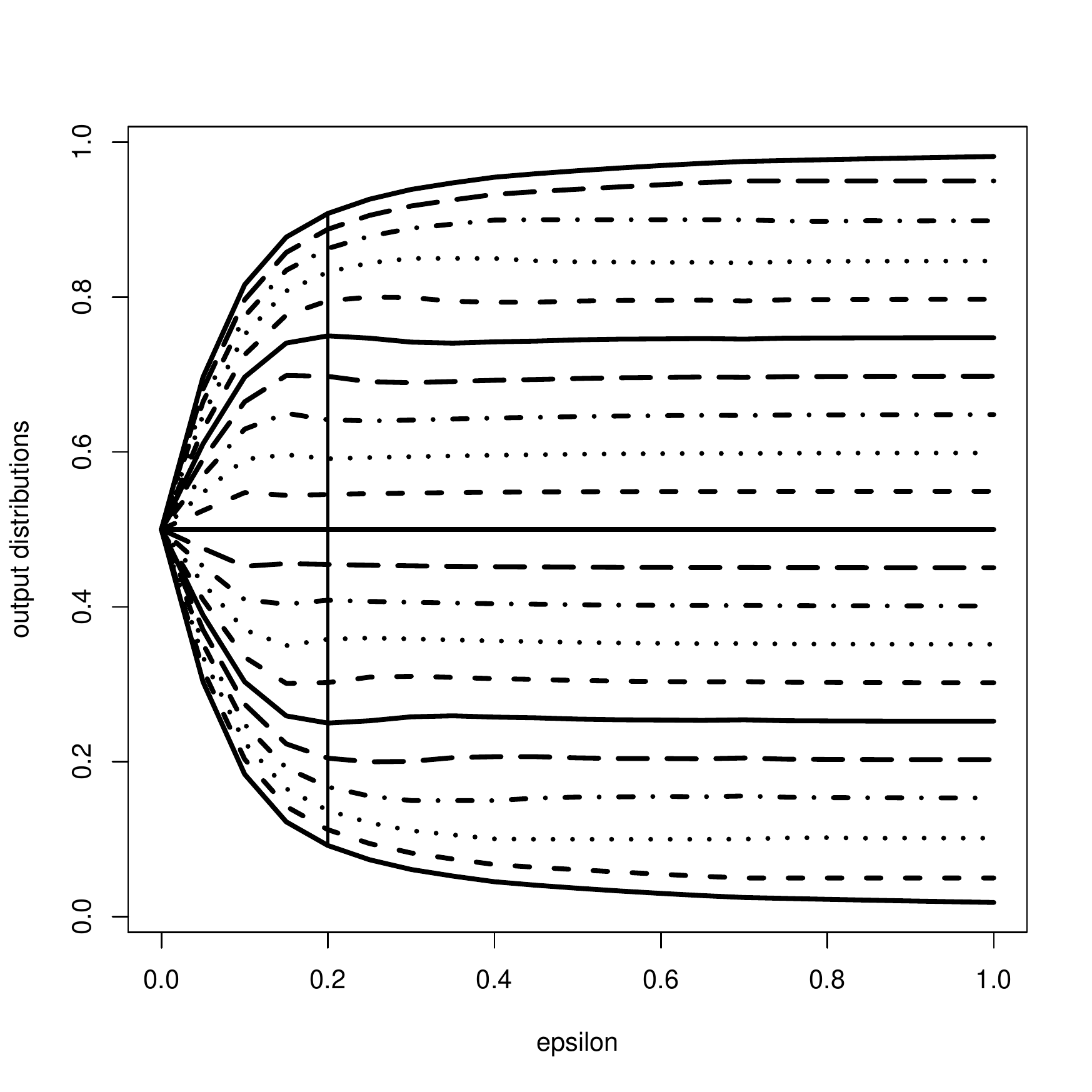}
	}
	}
	\caption{Probability that the output is $1$ as a function of $\epsilon$ for each profile for the Bernoulli-Chain experiments. The top row represents the baseline local differential privacy method, and the bottom row represents our method. The left column is Bernoulli-Chain-6, and the right column is Bernoulli-Chain-21. A higher spread along the interval $[0,1]$ means less noise added and is better. A vertical line has been drawn at $\epsilon=0.2$ to illustrate the spread.}
\label{fig:experiments2}

\end{figure}

Figures~\ref{subfig:catbaseline}-\ref{subfig:catourmethod} plot the utility across different outputs in the Categorical-Chain setting.  We illustrate its behavior through a small setting with 3 profiles, each with 4 categories. We can no longer plot the entirety of these profiles, so at each privacy level we measure the maximum absolute error for each output. Thus, in this setting, each privacy level is associated with 4 costs of the form given in \eqref{eq:catcost}. This permits the higher fidelity of profile-irrelevant information to be seen.

\begin{equation}
cost_j = \text{max}_{i \in [n]} | P^i_j - (P^iA^i)_j| \label{eq:catcost}
\end{equation}

Our experiments show the categories less associated with the profile identity have lower associated costs than the more informative ones. However, the local differential privacy baseline fails to exploit any of this structure and performs worse.

\section{Conclusion}

In conclusion, we provide a novel definition of local privacy -- profile based privacy -- that can achieve better utility than local differential privacy. We prove properties of this privacy definition, and provide mechanisms for two discrete settings. Simulations show that our mechanisms offer superior privacy-utility trade-offs than standard local differential privacy.

\subsection*{Acknowledgements.} We thank ONR under N00014-16-1-261, UC Lab Fees under LFR 18-548554 and NSF under 1804829 for research support.  We also thank Takao Murakami for pointing us to \cite{kawamoto2018differentially} and discussions about Observation~\ref{obs:parallelcomp}.

\bibliography{recentkc,robust,DPEPrefs,privacy2}
\bibliographystyle{abbrv}

\section{Proof of Theorems and Observations}

\begin{Obs}
If a data sample $X_i$ is drawn from profile $P_i$, and $\mathcal{A}$ preserves $(G,\epsilon)$-profile-based privacy, then for any (potentially randomized) function $F$, the release $F(\mathcal{A}(X_i,P_i))$ preserves $(G,\epsilon)$-profile-based privacy.
\end{Obs}

\begin{proof}

Let $X_i \sim P_i$ and $X_j \sim P_j$ represent two random variables drawn from the profiles $P_i$ and $P_j$. We define additional random variables as  $Y_i= \mathcal{A}(X_i,P_i)$ and $Z_i = F(\mathcal{A}(X_i,P_i))$, along with the corresponding $Y_j$ and $Z_j$ that use $X_j$ and $P_j$. This is a result following from a standard data processing inequality.

\begin{align}
\frac{Pr(Z_i = z)}{Pr(Z_j = z)} &= \frac{\int_\mathcal{Y}Pr(Z_i=z,Y_i=y)dy}{\int_\mathcal{Y}Pr(Z_j=z,Y_j=y)dy} \\
&= \frac{\int_\mathcal{Y}Pr(Z_i=z|Y_i=y)Pr(Y_i = y)dY}{\int_\mathcal{Y}Pr(Z_j=z|Y_j=y)Pr(Y_j = y)dY} \\
&\le \max_Y \frac{Pr(Y_i = y)}{Pr(Y_j = y)} \\
&\le e^{\epsilon}
\end{align}

\end{proof}

\begin{Obs}
If two independent samples $X_1$ and $X_2$ are drawn from profile $P_i$, and $\mathcal{A}_1$ preserves $(G,\epsilon_1)$-profile-based privacy and $\mathcal{A}_2$ preserves $(G,\epsilon_2)$-profile-based privacy, then the combined release $(\mathcal{A}_1(X_1,P_i), \mathcal{A}_2(X_2,P_i))$ preserves $(G, \epsilon_1 + \epsilon_2)$-profile-based privacy.

\end{Obs}

\begin{proof}
The proof of this statement relies on the independence of the two releases $Y_1 = \mathcal{A}_1(X_1,P_i)$ and $Y_2 = \mathcal{A}_2(X_2,P_i)$, given the independence of $X_1$ and $X_2$ from the same same profile $P_i$. Let $P_j$ be another profile such that there is an edge $(P_i,P_j)$ in $G$. We will introduce $X_1'$ and $X_2'$ as independent samples from the profile $P_j$, and define  $Y_1' = \mathcal{A}_1(X_1',P_j)$ and $Y_2' = \mathcal{A}_2(X_2',P_j)$ By marginalizing over the two independent variables $X_1, X_2$, we may bound the combined privacy loss. For brevity, we will use $Pr(X)$ as a shorthand for the density at an arbitray point $Pr(X=x)$.

\begin{align}
\frac{Pr(Y_1,Y_2)}{Pr(Y_1',Y_2')} &= \frac{\int_\mathcal{X} \int_\mathcal{X} Pr(X_1,Y_1,X_2,Y_2) dX_1 dX_2}{\int_\mathcal{X} \int_\mathcal{X} Pr(X_1',Y_1',X_2',Y_2'|P_j) dX_1 dX_2} \\
&= \frac{\int_\mathcal{X} \int_\mathcal{X} Pr(X_1,Y_1)Pr(X_2,Y_2) dX_1 dX_2}{\int_\mathcal{X} \int_\mathcal{X} Pr(X_1',Y_1')Pr(X_2',Y_2') dX_1 dX_2} \\
&= \frac{\int_\mathcal{X} Pr(X_1,Y_1) dX_1 \int_\mathcal{X}Pr(X_2,Y_2) dX_2}{\int_\mathcal{X} Pr(X_1',Y_1') dX_1 \int_\mathcal{X} Pr(X_2',Y_2') dX_2} \\
&= \frac{ Pr(Y_1) Pr(Y_2)}{Pr(Y_1')Pr(Y_2')} \\
& \le e^{\epsilon_1}e^{\epsilon_2}
\end{align}

\end{proof}

\begin{Obs}
This proof does not hold if $X_1$ and $X_2$ are not independent samples from $P_i$. This may occur if the same observational data $X$ is privatized twice, or if other correlations exist between $X_1$ and $X_2$. We do not provide a composition result for this case.
\end{Obs}

\begin{Obs}
If two profiles $P_i$ and $P_j$ are independently selected, and two observations $X_i \sim P_i$ and $X_j \sim P_j$ are drawn,  and $\mathcal{A}_1$ preserves $(G,\epsilon_1)$-profile-based privacy and $\mathcal{A}_2$ preserves $(G,\epsilon_2)$-profile-based privacy, then the combined release $(\mathcal{A}_1(X_i,P_i), \mathcal{A}_2(X_j,P_j))$ preserves $(G, \text{max}\{\epsilon_1,\epsilon_2\})$-profile-based privacy.

\end{Obs}

\begin{proof}
For the purposes of this setting, let $Q_1$ and $Q_2$ be two random variables representing the choice of profile in the first and second selections, with the random variables $X_1 \sim Q_1$ and $X_2 \sim Q_2$.

Since the two profiles and their observations are independent, the two releases $Y_1 = \mathcal{A}_1(X_1,Q_1)$ and $Y_2 = \mathcal{A}_2(X_2,Q_2)$ contain no information about each other. That is, $Pr(Y_1=y_1|Q_1=P_i,Q_2 = P_j, Y_2=y_2) = Pr(Y_i|Q_1 = P_i)$. Similarly we have $Pr(Y_2=y_2|Q_1 = P_i,Y_1=y_1,Q_2=P_j) = Pr(Y_2=y_2)$.

Let $P_h$ and $P_k$ be profiles such that the edges $(P_h,P_i)$ and $(P_j,P_k)$ are in $G$.

\begin{align}
\frac{Pr(Y_1,Y_2|Q_1=P_i)}{Pr(Y_1,Y_2|Q_1=P_h)} &= \frac{\sum_{Q_2} Pr(Y_1,Y_2,Q_2|Q_1=P_i)}{\sum_{Q_2} Pr(Y_1,Y_2,Q_2|Q_1=P_h)} \\
&= \frac{\sum_{Q_2} Pr(Y_1|Q_1=P_i,Y_2,Q_2)Pr(Y_2,Q_2|Q_1=P_i)}{\sum_{Q_2} Pr(Y_1|Q_1=P_h,Y_2,Q_2)Pr(Y_2,Q_2|Q_1=P_h)} \\
&= \frac{\sum_{Q_2} Pr(Y_1|Q_1=P_i)Pr(Y_2,Q_2)}{\sum_{Q_2} Pr(Y_1|Q_1=P_h)Pr(Y_2,Q_2)} \\
&= \frac{Pr(Y_1|Q_1=P_i)}{Pr(Y_1|Q_1=P_h)} \cdot \frac{\sum_{Q_2} Pr(Y_2,Q_2)}{\sum_{Q_2}Pr(Y_2,Q_2)} \\
&= \frac{Pr(Y_1|Q_1=P_i)}{Pr(Y_1|Q_1=P_h)}\\
&\le e^{\epsilon_1}
\end{align}

A similar derivation conditioning on $P_2=P_j$ results in a ratio bounded by $e^{\epsilon_2}$ over the edge $(P_j,P_k)$. Thus to get a single bound for the combined release $(Y_1,Y_2)$ over the edges of the graph $G$, we take the maximum $e^{\text{max}\{\epsilon_1,\epsilon_2\}}$.

\end{proof}

\begin{Obs}
This proof does not hold if the profile selection process is not independent. We do not provide a composition result for this case.

\end{Obs}

\begin{Thm}

The One Bit Cluster mechanism achieves $\epsilon$-profile based privacy.

\end{Thm}

\begin{proof}
	By direct construction, it is known that the flipping probabilities generated for single edges $\alpha_e$ will satisfy the privacy constraints. What remains to be shown for the privacy analysis is that taking  $\alpha = \max_{e} \alpha_e$ will satisfy the privacy constraints for all the edges simulataneously.
	
	To show this, we will demonstrate a monotonicity property: if a flipping probability $\alpha < 0.5$ guarantees a certain privacy level, then so too do all the probalities in the interval $(\alpha,0.5)$. By taking the maximum across all edges, this algorithm exploits the monotonicity to ensure all the constraints are met simultaneously.
	
	Let $x_1\sim P_1$ and $x_2 \sim P_2$, and let $p_i$ and $p_j$ denote the parameters of these two Bernouli distributions. When computing the privacy level, we have two output values and thus two output ratios to consider:
	
	\begin{align}
	\left|\log\frac{Pr(\mathcal{A}(x_1,P_1)=1)}{Pr(\mathcal{A}(x_2,P_2)=1)} \right| &= \left| \log \frac{ p_1\cdot (1-\alpha) + (1-p_1) \cdot \alpha}{p_2\cdot (1-\alpha) + (1-p_2)\cdot \alpha} \right|\\
	\left|\log\frac{Pr(\mathcal{A}(x_1,P_1)=0)}{Pr(\mathcal{A}(x_2,P_2)=0)} \right| &= \left| \log \frac{ p_1\cdot \alpha + (1-p_1) \cdot (1-\alpha)}{p_2\cdot \alpha + (1-p_2)\cdot (1-\alpha)} \right|
	\end{align}
	
	Without loss of generality, assume $p_1 > p_2$. (If they are equal, then all possible privacy levels are achieved trivially.) This means following two quantities are positive and equal to the absolute values above when $\alpha < 0.5$.
	
	\begin{align}
	\log\frac{Pr(\mathcal{A}(x_1,P_1)=1)}{Pr(\mathcal{A}(x_2,P_2)=1)} &=  \log \frac{ p_1\cdot (1-\alpha) + (1-p_1) \cdot \alpha}{p_2\cdot (1-\alpha) + (1-p_2)\cdot \alpha} \\
	\log\frac{Pr(\mathcal{A}(x_1,P_1)=0)}{Pr(\mathcal{A}(x_2,P_2)=0)}  &=  \log \frac{ p_2\cdot \alpha + (1-p_2) \cdot (1-\alpha)}{p_1\cdot \alpha + (1-p_1)\cdot (1-\alpha)}
	\end{align}
	
	Our next task is to show that these quantities reveal monotonic increases in privacy levels as $\alpha$ increases up to 0.5. Taking just the $\mathcal{A}(x_1,P_1)=1$ term for now, we compute the derivatives.
	
	\begin{align}
	\frac{\partial}{\partial \alpha} \left[ \log \frac{ p_1\cdot (1-\alpha) + (1-p_1) \cdot \alpha}{p_2\cdot (1-\alpha) + (1-p_2)\cdot \alpha} \right]&= \frac{ 1-2p_1}{p_1\cdot (1-\alpha) + (1-p_1) \cdot \alpha} - \frac{ 1-2p_2}{p_2\cdot (1-\alpha) + (1-p_2) \cdot \alpha} \\
	&\hskip-60pt = \frac{ 1-2p_1}{p_1 + (1-2p_1)\alpha} - \frac{ 1-2p_2}{p_2 + (1-2p_2)\alpha} \\
	&\hskip-60pt= \frac{ \left(1-2p_1\right) \left(p_2\cdot (1-\alpha) + (1-p_2) \cdot \alpha\right) - \left( 1-2p_2\right)\left(p_1\cdot (1-\alpha) + (1-p_1) \cdot \alpha\right) }{\left(p_1\cdot (1-\alpha) + (1-p_1) \cdot \alpha\right)\left(p_2\cdot (1-\alpha) + (1-p_2) \cdot \alpha\right)} \\
	&\hskip-60pt= \frac{p_2 - p_1}{Pr(\mathcal{A}(x_1,P_1)=1)Pr(\mathcal{A}(x_2,P_2)=1)} \leq 0
	\end{align}
	
	The final inequality arises from our assumption that $p_1 > p_2$. A similar computation on the $\mathcal{A}(x_1,P_1)=0$ term also finds that the derivative is always negative.
	
	This monotonicity implies that increasing $\alpha$ (up to 0.5 at most) only decreases the probability ratios. In the limit when $\alpha = 0.5$, the ratios are precisely 1 (and the logarithm is 0). Thus if $\alpha < 0.5$ achieves a certain privacy level, all $\alpha'$ satisfying $\alpha < \alpha' < 0.5$ achieve a privacy level at least as strong.
	
	The One Bit Cluster mechanism takes the maximum across all edges, ensuring the final flipping probability is no less than the value needed by each edge to achieve probability ratios within $e^{\pm\epsilon}$. Therefore each edge constraint is satisfied by the final choice of flipping probability, and the mechanism satisfies the privacy requirements. 
	
\end{proof}

\begin{theorem}
The Smooth One Bit mechanism achieves $(G,\epsilon)$-profile-based privacy.
\end{theorem}

\begin{theorem}
The Smooth Categorical mechanism achieves $(G,\epsilon)$-profile-based privacy.
\end{theorem}

The Smooth One Bit mechanism and Smooth Categorical mechanism satisfy a privacy analysis directly from the constraints of the optimization problem. These optimizations are done without needing any access to a sensitive observation, and as such pose no privacy risk. Implicitly, the solution to the optimization problem is verified to satisfy the constraints before being used.

 \begin{theorem}
 If $\mathcal{A}$ is a mechanism that preserves $\epsilon$-local differential privacy, then for any graph $G$ of sensitive profiles, $\mathcal{A}$ also preserves $(G,\epsilon)$-profile-based differential privacy.
 \end{theorem}
 
 \begin{proof}
The proof of this theorem lies in that the strong protections given by local differential differential privacy to the observed data also extend to protecting the profile identities. Let $Y_i= \mathcal{A}(X_i,P_i)$, the output of a locally differentially private algorithm $\mathcal{A}$ that protects any two distinct data observations $x$ and $x'$. As local differential privacy mechanisms do not use profile information, the distribution of $Y_i$ depends only on $X_i$ and ignores $P_i$. To prove the generality of this analysis over any graph $G$, we will show the privacy constraint is satisfied for any possible edge $(P_i,P_j)$ of two arbitrary profiles.

\begin{align}
\frac{Pr(Y_i=y)}{Pr(Y_j=y)} &= \frac{\int_\mathcal{X}Pr(Y_i=y|X_i=x)Pr(X_i=x)dx}{\int_\mathcal{X}Pr(Y_i=y|X_j=x)Pr(X_j=x)dx} \\
&\le \frac{\sup_x Pr(Y_i = y | X_i = x)}{\inf_X Pr(Y_j = y|X_j = x)} \\
&\le e^{\epsilon}
\end{align}

If the final inequality did not hold, one would be able to find two values $X$ and $X'$ such that the output $Y$ violates the local differential privacy constraint, which contradicts our assumption on $\mathcal{A}$.

 \end{proof}
 
 \begin{Obs} \label{thm:appendixutility}
Suppose we are in the single-bit setting with two Bernoulli profiles $P_i$ and $P_j$ with parameters $p_i$ and $p_j$ respectively. If $p_i \leq p_j \leq e^{\epsilon} p_j$, then the solution $\alpha$ to~\eqref{eq:onebitprob} satisfies 
\begin{equation}
\alpha \leq\max\{0,\frac{p_j - e^{\epsilon}p_i}{2(p_j - e^{\epsilon}p_i) - (1-e^{\epsilon})}, \frac{p_i - e^{\epsilon} p_j + e^{\epsilon} - 1}{2(p_i - e^{\epsilon} p_j) +  e^{\epsilon} - 1}\}\mbox{.}
\end{equation}

\end{Obs} 

\begin{proof}

Direct computation shows the desired constraints are met with this value for $\alpha$.

\begin{eqnarray}
\min && \alpha  \label{eq:appendixonebitprob} \\
{\text{subject\ to }} && \alpha \geq 0 \nonumber\\
&&\frac{p_i (1 - \alpha) + (1 - p_i) \alpha}{p_j (1 - \alpha) + (1 - p_j) \alpha} \in [e^{-\epsilon}, e^{\epsilon}] \nonumber\\
&&\frac{(1 - p_i) (1 - \alpha) + p_i \alpha }{(1 - p_j) (1 - \alpha) + p_j \alpha} \in [e^{-\epsilon}, e^{\epsilon}] \mbox{.}\nonumber
\end{eqnarray}

First, we note that by our assumption $p_i \leq p_j$ and $\epsilon \geq 0$, we immediately have two of our constraints trivially satisfied given $\alpha \leq 0.5$, since $p_i (1 - \alpha) + (1 - p_i) \alpha \leq p_j (1 - \alpha) + (1 - p_j) \alpha$ and $(1 - p_i) (1 - \alpha) + p_i \alpha \geq (1 - p_j) (1 - \alpha) + p_j \alpha$.

Two constraints of interest remain: 

\begin{eqnarray}
&&\frac{p_i (1 - \alpha) + (1 - p_i) \alpha}{p_j (1 - \alpha) + (1 - p_j) \alpha} \geq e^{-\epsilon} \label{eq:firstconstr}\\
&&\frac{(1 - p_i) (1 - \alpha) + p_i \alpha }{(1 - p_j) (1 - \alpha) + p_j \alpha} \leq e^{\epsilon} \mbox{.}\label{eq:secconstr}
\end{eqnarray}

We know that these ratios are monotonic in $\alpha$, so to solve these inequalities, it suffices to find the values of $\alpha$ where we have equality on these two constraints. Any values of $\alpha$ larger than this (and less than $1/2$) will therefore satsify the inequality.

For \eqref{eq:firstconstr}, we get $\alpha = \frac{p_j - e^{\epsilon}p_i}{2(p_j - e^{\epsilon}p_i) - (1-e^{\epsilon})}$. Solving \eqref{eq:secconstr} instead, we get $\alpha = \frac{p_i - e^{\epsilon} p_j + e^{\epsilon} - 1}{2(p_i - e^{\epsilon} p_j) +  e^{\epsilon} - 1}$.

Since both constraints must be satisfied simultaneously, we can complete our statement by taking the maximum of the two points given by our constraints, along with knowing $\alpha \geq 0$.

\end{proof}

\end{document}